\documentclass[journal, onecolumn, 11pt]{IEEEtran}
\usepackage{comment, amsmath,amssymb, color,
  subfigure, tikz, graphicx, amsthm}
\usepackage{cite}
\usepackage{xifthen}

\newcommand{\demandSymbol}{x}

\newcommand{\cui}[1]{#1}
\newcommand{\cuit}[1]{#1}
\newcommand{\cuic}[1]{}
\newcommand{\cuid}[1]{#1}
\newcommand{\cuie}[1]{#1}
\newcommand{\onlyJournalCand}[1]{#1}
\newcommand{\onlyJournal}[1]{#1}

\newcommand{\onlyConf}[1]{}

\newcommand{\demtfinal}[3]{\ensuremath{\hat{#3{\demandSymbol}}^{#1}_{#2}}}

\newcommand{\demt}[2]{  
  \ifthenelse{\isempty{#1}}
  {
    \ensuremath{\demtfinal{#1}{#2}{\mathbf}}
  }
  {
    \ifthenelse{\isempty{#2}}{
      \ensuremath{\demtfinal{#1}{#2}{\mathbf}}
    }
    {
      \ensuremath{\demtfinal{#1}{#2}{}}
    }
  }
}

\newcommand{\demfinal}[4]{\ensuremath{#3{#4}^{#1}_{#2}}}

\newcommand{\dem}[3][\demandSymbol]{
  \ifthenelse{\isempty{#2}}
  {
    \ensuremath{\demfinal{#2}{#3}{\mathbf}{#1}}
  }
  {
    \ifthenelse{\isempty{#3}}{
      \ensuremath{\demfinal{#2}{#3}{\mathbf}{#1}}
    }
    {
      \ensuremath{\demfinal{#2}{#3}{}{#1}}
    }
  }
}

\newcommand{\AvgCosMC}[2]{\dem[L]{#1}{#2}}

\newcommand{\dembfinal}[3]{\ensuremath{#3{z}^{#1}_{#2}}}

\newcommand{\demb}[2]{  
  \ifthenelse{\isempty{#1}}
  {
    \ensuremath{\dembfinal{#1}{#2}{\mathbf}}
  }
  {
    \ifthenelse{\isempty{#2}}{
      \ensuremath{\dembfinal{#1}{#2}{\mathbf}}
    }
    {
      \ensuremath{\dembfinal{#1}{#2}{}}
    }
  }
}

\newcommand{\demsum}[3][\demandSymbol]{\demfinal{#2}{#3}{}{#1}}

\newcommand{\numPlayer}{N}
\newcommand{\numTimeSlots}{T}
\newcommand{\demConstraint}[1]{\beta_{#1}}

\newcommand{\Dem}[2]{\mathcal{X}^{#1}_{#2}}
\newcommand{\payoff}[2]{\pi^{#1}_{#2}}
\newcommand{\revenue}[2]{E^{#1}_{#2}}
\newcommand{\rateofrevenue}[2]{\phi^{#1}_{#2}}
\newcommand{\payment}[2]{M^{#1}_{#2}}
\newcommand{\linkcost}[2]{J^{#1}_{#2}}
\newcommand{\linkcostdiff}[2]{K^{#1}_{#2}}

\newcommand{\maxWholeSaleCapacity}{U}

\newcommand{\totalFuelCost}{C}
\newcommand{\totalFuelCostMCP}{\mathcal{C}}
\newcommand{\avgCostPrice}{A}
\newcommand{\IBP}{\mathcal{M}}

\newtheorem{lem}{Lemma}
\newtheorem{defi}{Definition}

\begin{document}
\title{Noncooperative Games for Autonomous Consumer Load Balancing over Smart Grid}
\author{Tarun~Agarwal and~Shuguang~Cui 
  \thanks{The
  authors are with the Department of Electrical and Computer
  Engineering, Texas A\&M University, College Station, TX 77843-3128,
  USA. (Email: atarun@neo.tamu.edu, cui@ece.tamu.edu)}\onlyConf{\\
  Department of Electrical and Computer Engineering\\
  Texas A\&M University\\
  Email: atarun@neo.tamu.edu, cui@ece.tamu.edu}}

\maketitle

\begin{abstract}
  Traditionally, most consumers of electricity pay for their
  consumptions according to a fixed rate. With the advancement of Smart
  Grid technologies, large-scale implementation of variable-rate
  metering becomes more practical.  As a result, consumers will be
  able to control their electricity consumption in an automated
  fashion, where one possible scheme is to have each individual
  maximize its 
  own utility as a
  noncooperative game.  In this paper, noncooperative games are
  formulated among the \cui{electricity} consumers \cui{in} Smart Grid with two real-time
  pricing schemes, where the Nash equilibrium operation points are
  investigated for their uniqueness and load balancing properties. The
  first pricing scheme charges a price according to the average cost
  of electricity borne by the retailer and the second one charges
  according to a time-variant increasing-block price\cui{, where for each scheme, a} zero-revenue model and a \cui{constant-rate revenue} model are
  considered.  \onlyJournal{\cui{In addition,} the relationship between the studied
  games and certain \cuie{competitive routing games} 
  from
  the computer networking community, known as atomic flow games, is \cui{established, for which it is} shown that
  the proposed noncooperative game formulation falls under the class
  of atomic splittable flow games.}  The Nash equilibrium is shown to
  exist for four different \cui{combined} cases corresponding to \cui{the} two pricing schemes
  and \cui{the} two revenue models, and is 
  unique for three of the
  cases under certain conditions. It is \cui{further} shown that both pricing
  schemes lead to similar electricity loading patterns when consumers
  are only interested in minimizing the electricity costs \cui{without any other profit considerations}.
  Finally, the conditions under which the increasing-block pricing
  scheme is preferred over the average-cost based pricing scheme are
  discussed.
\end{abstract}


\bstctlcite{IEEEexample:BSTcontrol}

\section{Introduction}
\label{sec:introduction}

In the traditional power market, electricity consumers usually
pay a
fixed 
retail price for their electricity usage. This price
\cui{only} changes on a seasonal or yearly basis. However, it has been
long recognized in the economics community that charging consumers a
flat rate for electricity creates allocative inefficiencies, i.e.,
consumers do not pay equilibrium prices \cui{according to} 
their consumption \cui{levels} \cite{allcott2009rethinking}. 
This \cui{was} 
shown
through an example in \cite{borenstein2005time}, which illustrates how
flat pricing causes deadweight loss at off-peak times and excessive
demand 
at the peak times. The latter 
\cui{may lead to small-scale}
blackouts in \cui{a} short run and excessive capacity buildup over \cui{a}
long run. \cui{As a solution,} variable-rate metering that reflects the real-time cost of \cui{power}
generation can \cui{be used to} influence consumers to defer their power \cui{consumption away} from the
peak times. The reduced peak-load can significantly reduce the need
for expensive \cui{backup} generation during peak times and excessive \cui{generation} capacity.

The main technical hurdle in implementing 
real-time pricing has
been the lack of cost-effective \cui{two-way} smart metering, which can communicate
real-time prices to consumers and their consumption levels back to the
energy provider. \cui{In addition,} the claim of social benefits from real-time pricing
also assumes that the consumer demand is elastic and responds to price
changes \cui{while} 
traditional consumers do not possess the equipments that
enable them to quickly alter their demands \cui{according to the} 
changing \cui{power} prices\cui{. 
S}ignificant research efforts on real-time pricing have involved
estimating the consumer demand elasticity and the level of benefits
that real time pricing can achieve \cite{ allcott2009rethinking,
  holland2006short, borenstein2004long}.  
\cui{Fortunately, the above} requirements
\cui{on} smart metering and consumer adaptability 
\cui{are being fulfilled \cite{faruqui2010rethinking} as technology advances in}
cyber-enabled metering\cui{, power generation, power storage, and manufaturing automation, which is driven by the need for a  Smart Grid.}




\cuid{Such real-time pricing dynamics have been studied in the
  literature mainly with game theory \cite{fahrioglu1999designing,
    caron2010incentive, ibars2010distributed}. In particular, the
  authors in \cite{fahrioglu1999designing} provided a design mechanism 
  with \emph{revelation principle} to determine the optimal amount of
  incentive that is needed for the customers  to be willing
  to enter a contract with the utility and accept power curtailment
  during peak periods. However, they only considered a fixed pricing scheme.  
  In \cite{caron2010incentive}, the authors studied
  games among consumers under a 
  certain class of demand
  profiles at a price that is a function of day long aggregate cost of global electricity load
  of all consumers. However, the case with real-time prices was not investigated in \cite{caron2010incentive}.
  In \cite{ibars2010distributed}, a noncooperative game was studied to tackle the real-time pricing problem, where the solution was obtained by exploring the relationship with the \cuie{congestion games and} potential games. \cuie{However, the pricing schemes that we study are not amenable to transformations described in \cite{ibars2010distributed}.}}

In this paper we formulate noncooperative games
\cite{tirole1988theory, başar1999dynamic} among the consumers 
with two real-time pricing schemes under more general load profiles and revenue models. The first pricing scheme
charges a price according to \cui{the} instantaneous average cost of electricity
production and the second one charges according to a time-varying
version of increasing-block price \cite{borenstein2008equity}. We
investigate consumer demands at the Nash equilibrium operation points
for their uniqueness and load balancing properties. Furthermore, two revenue models
are considered for each of the schemes, and we show that both pricing schemes lead to similar
electricity loading patterns when consumers are interested only in the
minimization of electricity costs. 
\onlyJournal{We \cui{also} demonstrate the relationship between these games and certain
competitive routing games \cite{orda1993competitive}, known as atomic flow
games \cite{roughgarden2005selfish} from the computer networking
community. We show that the proposed noncooperative game formulation
falls under the class of atomic splittable flow games
\cite{bhaskar2009equilibria}. Specifically, we show that the
noncooperative game amongst the consumers has the same structure as \cui{that} in
the atomic splittable flow game over a two-node network with multiple
parallel links between them. }Finally we discuss the conditions under
which the increasing-block pricing scheme is preferred over the
average-cost based pricing scheme.

The \cui{rest of the} paper is organized as follows. The system model\onlyJournal{,}\onlyConf{~and} formulation of
the noncooperative game\onlyJournal{, and its relationship 
\cui{with the} atomic flow games} are
presented in Section~\ref{sec:system-model}. The game is analyzed
with different real-time pricing schemes under different revenue
models in Sections~\ref{sec:nash-equil-with} and~\ref{sec:new-pricing-scheme}, where the Nash equilibrium
properties are investigated. We conclude the paper in
Section~\ref{sec:conclusion}.



\section{System Model\onlyConf{~and Game Formulation}}
\label{sec:system-model}
\onlyConf{\subsection{System Model}}
We study the transaction of energy between a single electricity
retailer and multiple consumers. 
In \cui{each} given time slot, each consumer has a demand for electric energy (measured in
Watt-hour, Wh).  The job of the retailer is to
satisfy demands from all the consumers.
The electricity supply of the retailer is purchased from a variety of
sources over a wholesale electricity market and the retailer may
possess some generation capacity \cui{as well}. These
sources may use different technologies and fuels to generate
electricity, which leads to different marginal costs of electricity
at the retailer, where 
the marginal cost is the incremental cost incurred
to produce an additional unit of output \cite{lindeman2001ez}.
Mathematically, the marginal cost function  is
expressed as the first derivative of the total cost function. Examples of \cui{the} marginal cost function and \cui{the corresponding} total cost are presented in
Fig.~\ref{fig:mcgs1} and Fig.~\ref{fig:mcgs2}, respectively, \cui{which are} 
based on real world data from \cui{the} wholesale electricity market 
\cite{holland2006short}. Naturally, the retailer attempts to
satisfy demands by procuring the cheapest source first\onlyJournal{\footnote{In
  real life the base load, i.e., the regular power that is demanded by
  the consumers, is satisfied from sources such as hydro, coal or
  nuclear, as they are cheap. The fluctuating components of the
  demand 
  are satisfied from sources such as oil, as the power-plants based on oil
  are more flexible to control.
}}. This results in a non-decreasing marginal cost of the supply curve,
as illustrated through the example in Fig.~\ref{fig:mcgs1}. The
retailer charges each consumer a certain price for its consumption
in order to cover the cost, where
the sum payments by all the consumers should be enough to
cover the \cui{total} 
\cui{cost and certain} profit margin \cui{set by the retailer or regulatory body}. In our model we assume that all these 
are
incorporated within the marginal cost of electricity.
  \begin{figure}[!h]
    \centering
    \subfigure[Marginal cost  as a function of
    supply]{\begin{tikzpicture}[scale = 0.5, domain=0:12, samples=100]
        \draw[thick,color=gray,step=4cm, dashed] (0,0) grid (12,12);
        \draw[->] (-1,0) -- (12.5,0) node[midway, sloped, below]
        {Quantity Supplied (MWh)};
\draw[->] (0,-1) -- (0,12.5)
node[midway, sloped, above] {Marginal Cost (\$/MWh)};
\draw plot[id=xwok] function{1/(1+exp(-10*x))+1/(1+exp(-10*(x-1)))+1/(1+exp(-10*(x-4)))+1/(1+exp(-10*(x-7)))+2/(1+exp(-10*(x-10)))+2/(1+exp(-10*(x-11)))+exp(x/12)};
      \draw (12,0) node[below] {60000};
      \draw (0,12) node[left] {\$100};

      \draw (0.4,2.3)--(0.4,2.3) node[midway, sloped] {Hydro};
      \draw (2.3,3.2)--(2.3,3.2) node[midway, sloped] {Nuclear};
      \draw (5.4,4.5)--(5.4,4.5) node[midway, sloped] {Coal};
      \draw (8.2,5.9)--(8.2,5.9) node[midway, sloped] {Natural Gas};
      \draw (10.4,8.3)--(10.4,8.3) node[midway, sloped] {Oil};
      \draw (11.5,10.6)--(11.5,10.6) node[midway, sloped] {Oil};

\end{tikzpicture}
    \label{fig:mcgs1}}
\subfigure[Total cost as a function of supply]{\begin{tikzpicture}[scale = 0.5, domain=0:12, samples=100]
\draw[thick,color=gray,step=4cm,
dashed] (0,0) grid (12,12);
\draw[->] (-1,0) -- (12.5,0)
node[midway, sloped, below] {Quantity Supplied (MWh)};
\draw[->] (0,-1) -- (0,12.5)
node[midway, sloped, above] {Cost (\$)};
\draw plot[id=xwok2] function{ (1/10*log(-exp(10*x)-1)+1/10*log(-exp(10*x)-exp(10))+1/10*log(-exp(10*x)-exp(40))+1/10*log(-exp(10*x)-exp(70))+2/10*log(-exp(10*x)-exp(100))+2/10*log(-exp(10*x)-exp(110))+12*exp(x/12) - (1/10*log(-exp(10*0)-1)+1/10*log(-exp(10*0)-exp(10))+1/10*log(-exp(10*0)-exp(40))+1/10*log(-exp(10*0)-exp(70))+2/10*log(-exp(10*0)-exp(100))+2/10*log(-exp(10*0)-exp(110))+12*exp(0/12)))/2};
      \draw (12,0) node[below] {60000};
      \draw (0,12) node[left] {\$1M};
\end{tikzpicture}
    \label{fig:mcgs2}}

  \caption{A hypothetical marginal cost of supply and the
    corresponding total cost curve as seen by the retailer in the
    wholesale market within a single time slot. Supply is from five
    different sources: hydroelectric, nuclear, coal, natural gas, and
    oil. Two different generators may use different technologies for
    power generation thus incurring different marginal costs with the
    same fuel (e.g., the two different cost levels for oil in
    Fig.~\ref{fig:mcgs1}).}
    \label{fig:mcgs}
  \end{figure}
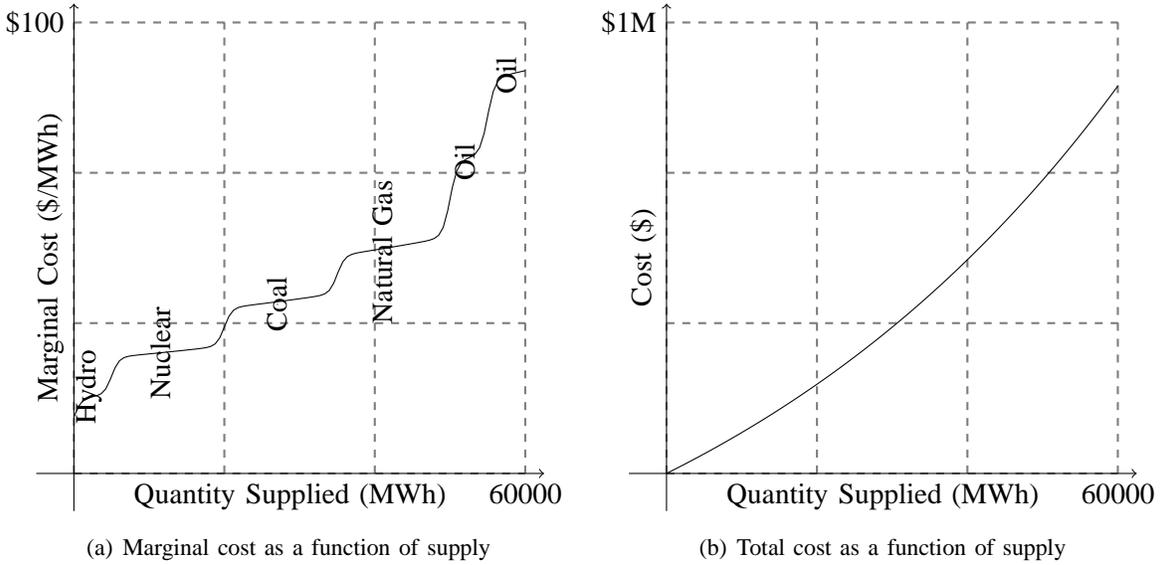

  While the retailer aims to procure sufficient supply to meet the sum
  demand of its consumers in each time slot, in reality, the supply is
  limited by the generation capacity available in the wholesale
  electricity market. Thus, the maximum sum load that the retailer can
  service \cui{bears} an upper limit and we model this capacity limit by
  setting the marginal cost of electricity to infinity when the sum load
  exceeds a predetermined threshold.
  Each consumer has an energy demand in 
\cui{each}  time slot and it pays
  the retailer \cui{at a} 
  price \cui{that is set by the retailer} 
  such that, in each time slot, the sum of payments made by all
  consumers meets the total cost in that slot.  \cui{As such,} a
  particular consumer's share of this bill depends on the retailer's
  pricing scheme, which is a function of the demands from all the
  consumers. Accordingly, as the total load varies over time, each
  consumer operates over a time-variant price with time-slotted
  granularity. \cui{We assume that} each consumer \cui{has a} 
  total demand
  for electricity over each day\footnote{Here we adopt one day as an operation period that contains a certain number of time slots. Obviously, such a choice has no impact on the analytical results in this paper.}, which can be distributed throughout
  the day in a time-slotted manner, to maximize certain utility
  function. \cui{Next,} we model such individual load balancing behaviors as a
  noncooperative game.

\subsection{Noncooperative Load Balancing Game}
\label{sec:nonc-game-form}

The noncooperative game between these consumers is formulated as
follows. Consider a group of $\numPlayer$ consumers, who submit their
daily demands to a retailer in a time-slotted pattern at the beginning
of the day (which contains $\numTimeSlots$ time slots).  These
consumers are selfish and aim to maximize their \cui{individual} 
utility/payoff \cui{functions;} hence they do not cooperate with each other to manage
their demands. Each consumer $i$ has a minimum total daily requirement
of energy, $\demConstraint{i}\geq 0$, which is split over the
$\numTimeSlots$ time slots. Let $\dem{i}{t}$ denote the $i$th
consumer's demand in the $t$th time slot.  A consumer can demand any
value $\dem{i}{t}\geq 0$ (negativity constraint) with
$\sum_t\dem{i}{t} \geq\demConstraint{i}$ (demand constraint).
Let
  \onlyJournal{\[\dem{i}{}=\{\dem{i}{1}, \dem{i}{2} ,\ldots,
  \dem{i}{t}, \ldots, \dem{i}{\numTimeSlots} \},\]}\onlyConf{$\dem{i}{}=\{\dem{i}{1}, \dem{i}{2} ,\ldots,
  \dem{i}{t}, \ldots, \dem{i}{\numTimeSlots} \}$,} represent the $i$th
  consumer's demand vector, which is called the strategy for the $i$th
  consumer. Let
    \onlyJournal{\[\dem{}{t} = \{
    \dem{1}{t} ,\ldots, \dem{\numPlayer}{t}\}, \]}\onlyConf{$\dem{}{t} = \{
    \dem{1}{t} ,\ldots, \dem{\numPlayer}{t}\}$,} represent the demand
    vector from all consumers in time slot $t$ \cui{with $\demsum{}{t} = \sum_i\dem{i}{t}$}. Let $\dem{}{}$
    represent the set $\{\dem{1}{}, \ldots, \dem{\numPlayer}{}\}$. 


    The payoff or utility for consumer $i$ is denoted by
    $\payoff{i}{}$ which is the 
    difference between the total revenue it generates from the
    purchased electricity and its cost.
    In particular,
    let $\revenue{i}{t}$, a function of $\dem{i}{t}$, represent the revenue generated by the $i$th
    consumer in the $t$th time
    slot
    and $\payment{i}{t}$, a function of $\dem{}{t}$, represent its payment to the retailer for
    purchasing $\dem{i}{t}$.
    Then the payoff $\payoff{i}{}$, to be
    maximized by consumer $i$, is given by \[\payoff{i}{} =
    \sum_{t\in\{1, \ldots, \numTimeSlots\}} \left(\revenue{i}{t} -
      \payment{i}{t} \right).\] Since $\payment{i}{t}$ is a function
    of $\dem{}{t}$, we see that the consumer payoff is influenced by
    its load balancing strategy and those of other consumers. 

    \cui{We consider the problem} of maximizing 
    the payoff at each consumer by designing the
    distributed load balancing strategy $\dem{i}{}$'s, \cui{under} 
    two
    real-time pricing schemes \cui{set by the retailer}. The first \cui{one} is the average-cost based
    pricing scheme and the second \cui{one} is the increasing-block
    pricing scheme. Specifically, for the first scheme the retailer
    charges the consumers the average cost of electricity
    procurement \cui{that is only dependent on the sum demands, $\demsum{}{t}$, from all the consumers}. For the second scheme, the retailer charges according to
    a  marginal cost function that depends on the vector of
    demands from all consumers, $\dem{}{t}$. 

Let $\totalFuelCost(\demsum{}{})$
represent the cost of $\demsum{}{}$ units of electricity, to the
retailer, \cui{from} 
the wholesale market
    (an example function is plotted in
    Fig.~\ref{fig:mcgs2}).
    Then under the average-cost based pricing, the price per unit
    charged to the consumers is given by
\begin{equation}
  \label{eq:2}
\onlyJournal{  \avgCostPrice(\demsum{}{t}) = \frac{\totalFuelCost(\demsum{}{t})}{\demsum{}{t}},}
\onlyConf{  \avgCostPrice(\demsum{}{t}) = \totalFuelCost(\demsum{}{t})/\demsum{}{t}},
\end{equation}
and at time $t$ consumer $i$ pays
\begin{equation}
\payment{i}{t}=\dem{i}{t}\avgCostPrice(\demsum{}{t})\label{eq:cost1}
\end{equation}
for consuming $\dem{i}{t}$ units of electricity. It is easy to see
that $\sum_i\payment{i}{t} = \totalFuelCost(\demsum{}{t})$, i.e., with
average-cost based pricing the total payment made by the consumers
covers the total cost to the retailer. Note that $\totalFuelCost'(\demsum{}{t})$
gives the marginal cost function in the wholesale market, henceforth
denoted by $\totalFuelCostMCP(\demsum{}{t}) = \totalFuelCost'(\demsum{}{t})$ in the context of
increasing-block pricing (an example marginal cost
curve is plotted in Fig.~\ref{fig:mcgs1}). For reasons we discussed
earlier, in the context of electricity market, the marginal cost
$\totalFuelCostMCP(\demsum{}{t})$ is always non-negative and
non-decreasing such that $\totalFuelCost(\demsum{}{t})$ is always
positive, non-decreasing, and
convex. Briefly, we note that as \cui{the retailer} capacity is constrained by a
predetermined upper limit $\maxWholeSaleCapacity$, we model this
constraint as $\totalFuelCost(\demsum{}{t})=\infty,~\forall\demsum{}{t}>
  \maxWholeSaleCapacity$\cui{; obviously} $\dem{i}{t}\leq\maxWholeSaleCapacity$ is \cui{an implicit constraint} 
  on the demand
  $\dem{i}{t}$ \cui{for any rational consumer}.


The second scheme is a time-variant version of the increasing-block
pricing scheme \cite{borenstein2008equity}. With a typical increasing-block pricing scheme,
consumer $i$ is charged a certain rate $b_1$ for its first $z_1$ units
consumed, then charged rate $b_2~ (> b_1)$ for additional $z_2$ units,
and charged rate $b_3~ (> b_2)$ for additional $z_3$ units, and so
on. The $b$'s and $z$'s describe the marginal cost price for the
 commodity.  In our scheme we design a marginal cost function, \cui{which}
 retains the increasing nature of increasing-block pricing, such that it depends on
$\dem{}{t}$ and the function $\totalFuelCost(\cdot)$. 
Consumer $i$ pays an
amount determined by the marginal cost function $\IBP(\demandSymbol,
\dem{}{t})$, applicable to all consumers at time slot $t$. In
particular consumer $i$ pays

\begin{equation}
\payment{i}{t}=\int_0^{\dem{i}{t}}\IBP(\demandSymbol, \dem{}{t})d\demandSymbol\label{cost2}
\end{equation}
for consuming $\dem{i}{t}$ units of electricity where $\IBP(\cdot)$ is chosen
as \[\IBP(\demandSymbol, \dem{}{t}) =  \totalFuelCostMCP\left(\sum_j \min{(\demandSymbol,\dem{j}{t})}\right),\]
such that $\sum_i\payment{i}{t} = \totalFuelCost(\demsum{}{t})$ is
satisfied. \cui{An intuition behind this pricing scheme is to penalize consumers with relatively larger demands.}
\cui{Note that in this case, $\dem{i}{t}\leq \maxWholeSaleCapacity$ is implicitly assumed by letting $\totalFuelCost(\cdot)=\infty~\forall \dem{i}{t}> \maxWholeSaleCapacity$ and hence $\payment{i}{t}=\infty~\forall \dem{i}{t}> \maxWholeSaleCapacity$.}

For each of the two pricing schemes, we study two different revenue
models. For the first \cui{one} we set $\revenue{i}{t}$ as zero for all
consumers over all time slots, which leads to payoff maximization
being the same as cost minimization from the point of view of the
consumers. For the second \cui{one} we assign consumer $i$ a constant
revenue rate $\rateofrevenue{i}{t}$ at each time slot $t$, which gives
$\revenue{i}{t} = \rateofrevenue{i}{t} \dem{i}{t}$ and leads to payoff maximization
being the same as profit maximization.

\onlyJournalCand{
\subsection{Atomic Flow Games with Splittable Flows}
\label{sec:parall-with-atom}

The noncooperative game that we have formulated in the previous
section is related to the following problem in the network routing
literature \cite{orda1993competitive,
  roughgarden2005selfish}. Consider several agents each \cui{trying} 
to establish paths from a specific source node to some destination
node in order to transport a fixed amount of traffic. In the context
of Internet, \cui{each} agent can be viewed as a manager of packet
routing. In the context of transportation, \cui{each} agent \cui{can be} a company
routing its fleet vehicles across the network of roads. The problem
here is of competitive routing between agents, where each agent needs
to deliver a given amount of flow over the network from its designated
origin node to the corresponding destination node. An agent can choose
how to divide its flow amongst the available routes. On each link the
agents experience a certain delay. In the case of computer networks,
if many agents collectively route a large number of packets through a
particular link, the packets will experience larger delays; and beyond
a certain level, the link may even start dropping packets, resulting
in infinite delay.  Such a delay can be referred to as cost, which is a
function of the link congestion or the total flow through the
link. The cost of a path is the sum of the link costs along the route.

To show the relationship between our noncooperative consumer load
balancing problem and the above routing problem, we can reformulate
the load balancing problem into the following routing game over \cui{a network with two nodes and} 
multiple links 
\cite{orda1993competitive}. We use notations similar to
\cui{there in} \cite{orda1993competitive} in the interest of readability. Let there be
$\numPlayer$ agents 
who share a common source
node and a common destination over a two-node network connected by
$\numTimeSlots$ parallel links (see
Fig.~\ref{fig:twonode}).  
\begin{figure}[htp]
  \centering \subfigure[Flows from the $i$th
  agent.]{ \begin{tikzpicture}[>=latex, scale=2] \node[circle,draw] (s)
      at (0,0) {s}; \node[circle,draw] (t) at (2,0) {t};
      
      \draw[->] (s.north) .. controls (0,1) and (2,1) .. (t.north) node[midway,above]{$\dem{i}{1}$}node[midway, below]{$\vdots$};

      \draw[->] (s.south) .. controls (0,-1) and (2,-1) .. (t.south)node[midway,above]{$\dem{i}{\numTimeSlots}$};

      \draw[->] (s.east) -- (t.west) node[midway, below]{$\vdots$}
      node[midway, above]{$\dem{i}{t}$};
    \end{tikzpicture}} \subfigure[Sum of flows from all the
  agents.]{ \begin{tikzpicture}[>=latex, scale=2] \node[circle,draw]
      (s) at (0,0) {s}; \node[circle,draw] (t) at (2,0) {t};
      
      \draw[->] (s.north) .. controls (0,1) and (2,1) .. (t.north) node[midway,above]{$\demsum{}{1}$}node[midway, below]{$\vdots$};

      \draw[->] (s.south) .. controls (0,-1) and (2,-1) .. (t.south)node[midway,above]{$\demsum{}{\numTimeSlots}$};

      \draw[->] (s.east) -- (t.west) node[midway, below]{$\vdots$}
      node[midway, above]{$\demsum{}{t}$};
    \end{tikzpicture}} \subfigure[Cost for different links for the
  $i$th agent.]{ \begin{tikzpicture}[>=latex, scale=2]
      \node[circle,draw] (s) at (0,0) {s}; \node[circle,draw] (t) at
      (2,0) {t};
      
      \draw[->] (s.north) .. controls (0,1) and (2,1) .. (t.north) node[midway,above]{$\linkcost{i}{1}$}node[midway, below]{$\vdots$};

      \draw[->] (s.south) .. controls (0,-1) and (2,-1) .. (t.south)node[midway,above]{$\linkcost{i}{\numTimeSlots}$};

      \draw[->] (s.east) -- (t.west) node[midway, below]{$\vdots$}
      node[midway, above]{$\linkcost{i}{t}$};
    \end{tikzpicture}}
  \caption{A two node network with $\numTimeSlots$ links between
    source $s$ and destination $t$.}
  \label{fig:twonode}
\end{figure}
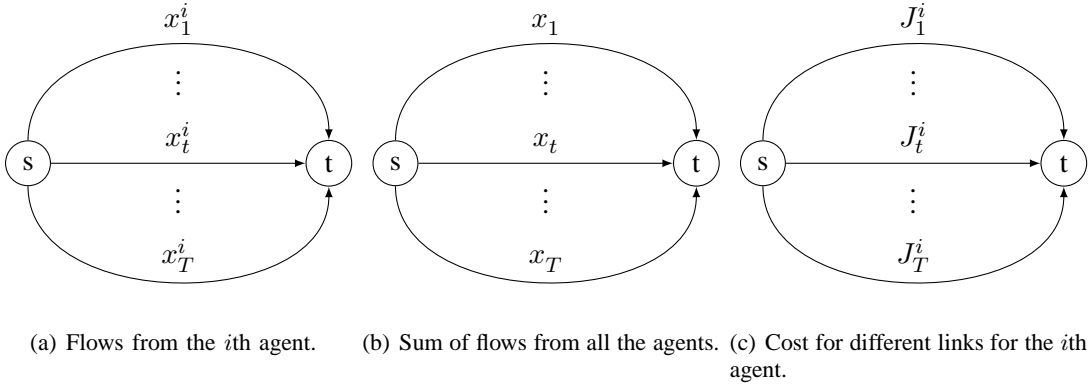
It is assumed that the agents do not cooperate. Each agent $i\in
\{1,\ldots,\numPlayer\}$ has a \cui{minimum} throughput demand $\demConstraint{i}$,
which can be split among the $\numTimeSlots$ links as chosen by the
agent.  
Let $\dem{i}{t}\geq 0$ denote the flow that
agent $i$ sends through link $t\in\{1,\ldots,\numTimeSlots\}$.  The
sum of $\dem{i}{t}$ should add upto $\demConstraint{i}$, i.e.,
$\demConstraint{i}=\sum_t\dem{i}{t}$
. Let $\demsum{}{t}=\sum_i\dem{i}{t}$, and $\dem{}{t} = \{\dem{1}{t},
\ldots, \dem{i}{t}, \ldots, \dem{\numPlayer}{t}\}$. The flow vector
for agent $i$ is denoted by the vector $\dem{i}{} = \{\dem{i}{1},
\ldots, \dem{i}{t}, \ldots, \dem{i}{\numTimeSlots}\}$. The system flow
vector is the collection of all agent flow vectors, denoted by
$\dem{}{}=\{\dem{1}{}, \ldots, \dem{i}{}, \ldots,
\dem{\numPlayer}{}\}$. A given $\dem{i}{}$ is feasible if its
components obey the non-negativity constraint and the demand
constraints. Let $\Dem{i}{}$ be the set of all feasible choices of
$\dem{i}{}$ for agent $i$, and $\Dem{}{}$ be the set of all feasible
choices of $\dem{}{}$.

Let $\linkcost{i}{}(\dem{}{})$ denote the \emph{cost} for each agent
$i$, \cui{which} it wishes to minimize.  \cui{Since} $\linkcost{i}{}(\dem{}{})$ is a
function of the flow vector of all the agents, the best response of a
given agent is a function of the responses \cui{from} all the \cui{other} agents; and
hence we can have a noncooperative game formulation. The Nash solution
of the game is defined as the system flow vector such that none
of the agents can unilaterally improve their performance. Formally,
$\demt{}{}\in \Dem{}{} $ is a Nash Equilibrium Point (NEP) if the
following condition holds for all agents
\begin{equation*}
  \begin{split}
    \linkcost{i}{}(\demt{}{})
    &=\min_{\dem{i}{}\in
      \Dem{i}{}}\linkcost{i}{}(\demt{1}{},
    \ldots, \demt{i-1}{}, \dem{i}{}, \demt{i+1}{}, \ldots,
    \demt{\numPlayer}{}),
  \end{split}
\end{equation*}
\cui{ where $\demt{j}{}$ is the demand vector for the $j$th agent.} The above noncooperative game is known as an \emph{atomic splittable
  flow game} \cite{bhaskar2009equilibria,
  roughgarden2005selfish2}. 
%
%
In \cite{orda1993competitive}, the existence of NEP is proved for \cui{the}
atomic splittable flow game over the two-node network with parallel
links \cui{when} the following five assumptions (G1-G5) are satisfied for the
cost function.

\begin{itemize}
\item[G1:] $\linkcost{i}{}$ is the sum of link cost functions, i.e.,
  $\linkcost{i}{}(\dem{}{}) =
  \sum_{t=1}^{\numTimeSlots}\linkcost{i}{t}(\dem{}{t})$.
\item[G2:] $\linkcost{i}{t}:[0,\infty)^{\numPlayer}\rightarrow
  [0,\infty]$ is a continuous function.
\item[G3:] $\linkcost{i}{t}$ is convex over $\dem{i}{t}$.
\item[G4:] Wherever finite, $\linkcost{i}{t}$ is continuously
  differentiable over $\dem{i}{t}$.
\item[G5:] Sum capacities of all links is greater than the sum 
  demands from all the agents.
\end{itemize}
The last assumption (G5) mentioned here is a simplification of the
original assumption mentioned in \cite{orda1993competitive}, applicable
to two-node networks, \cui{while the} original form of the assumption applies to
more general networks. The consequence of this assumption is that,
at Nash equilibrium, all users \cui{incur finite link costs, i.e., $\linkcost{i}{t}<\infty,~\forall  i,t$}. 

As a side-note, in \cite{orda1993competitive}, the uniqueness of NEP is further
imposed for \cui{the} two-node network if the cost function $\linkcost{i}{t}$
additionally complies with the following assumptions:
\begin{itemize}
\item[A1:] $\linkcost{i}{t}$ is a function of two arguments, namely
  agent $i$'s flow on link $t$ and the total flow on that link, i.e.,
  $\linkcost{i}{t}(\dem{}{t})=\bar{\linkcost{i}{t}}(\dem{i}{t},\demsum{}{t})$.
\item[A2:] $\bar{\linkcost{i}{t}}$ is increasing over each of its two
  arguments.
\item[A3:] Let $\linkcostdiff{i}{j}=\frac{\partial
    \bar{\linkcost{i}{t}}}{\partial \dem{i}{t}}$. Wherever $\linkcost{i}{t}$
  is finite, $\bar{\linkcost{i}{t}}$ is finite, and $\linkcostdiff{i}{t} =
  \linkcostdiff{i}{t}(\dem{i}{t},\demsum{}{t})$ is strictly increasing
  in each of its two arguments.
\end{itemize}
In particular, functions that comply with the assumptions G1-G5 and
A1-A3 are referred to as \textit{type-A} functions in
\cite{orda1993competitive}. In the following sections we will apply
some of the results in \cite{orda1993competitive} to facilitate our
analysis over the noncooperative consumer load balancing game. The
cost functions in our formulation do not satisfy all 
of the
assumptions A1-A3, and hence we use other means to prove uniqueness of NEP.

\cui{With our load balancing problem} for each of the two pricing schemes \cui{in our game} 
two
different revenue models are studied to provide more design insights,
which leads to two different payoff structures. In the first \cui{model} the
revenue is set to zero, such that payoff maximization is 
cost minimization. In the second \cui{model}, the rate of revenue generation
at each consumer is set as a non-zero constant, such that payoff
maximization is profit maximization.

}


\section{Nash Equilibrium with Average-Cost Pricing}
\label{sec:nash-equil-with}

For the average-cost pricing, the payment to the retailer in slot $t$ by consumer
$i$ is given by (\ref{eq:cost1}).


\subsection{Zero-Revenue Model}

In this case the revenue is set to zero as \onlyJournal{\[\revenue{i}{t} = 0,\]}\onlyConf{$\revenue{i}{t} = 0$,} 
which results in payoff maximization being the same as cost
minimization for each consumer. Specifically, the payoff for consumer
$i$ is given by
\onlyJournal{\[\payoff{i}{} =  -\sum_t \payment{i}{t} .\]}\onlyConf{$\payoff{i}{} =  -\sum_t \payment{i}{t}$.} 
The consumer \cui{load balancing} problem for consumer $i$, for \cui{$i=1,\ldots,\numPlayer$}, is given by
the following optimization problem:
      \begin{equation*}
      \begin{split}
        \text{maximize}\quad & \payoff{i}{}(\dem{i}{})=
        -\sum_t\payment{i}{t}\\
        \text{subject to}\quad 
        &\payment{i}{t} = \dem{i}{t}\avgCostPrice(\demsum{}{t}),\quad\forall t,\\
        & \sum_t\dem{i}{t}
        \geq\demConstraint{i},\\
        & \demsum{}{t}=\sum_j\dem{j}{t},\quad\forall t,
\\
         &0\leq \dem{i}{t},
         \quad\forall  t.\\
      \end{split}
    \end{equation*}
As cost to the retailer becomes infinity whenever the total
demand \cui{goes beyond} the capacity threshold for the wholesale market, \cui{i.e., when}
\onlyJournal{\[\totalFuelCost(\demsum{}{t})=\infty \quad\forall\demsum{}{t}>\maxWholeSaleCapacity,\]}\onlyConf{$\totalFuelCost(\demsum{}{t})=\infty \quad\forall\demsum{}{t}>\maxWholeSaleCapacity$,} 
the price to consumers will become infinite and their payoff will 
go to negative
infinity. 
Thus any consumer facing \cui{an} infinite cost \cui{at a particular time slot}  can manipulate the
demand vector such that \cui{the} 
cost becomes finite, \cui{which is always feasible \onlyJournal{as
    long as assumption G5 holds}\onlyConf{under the assumption that sum load demand over all times slots is less than sum supply availability}}. This implies that,
at Nash equilibrium, sum demand $\demsum{}{t}$ will be less than \cui{the} 
capacity threshold $\maxWholeSaleCapacity,~\forall t$, \cui{which} allows for a redundant constraint
$\dem{i}{t}\leq\maxWholeSaleCapacity,~\forall i, t$,
as
$\dem{i}{t}\leq\sum_i\dem{i}{t}=\demsum{}{t}\leq\maxWholeSaleCapacity$.
\cui{Such a redundant but explicit constraint} in turn
makes the feasible region for $\dem{}{}$, \cui{denoted by} $\Dem{}{}$, finite and hence
compact. The compactness property \onlyJournal{\cui{will be later}}\onlyConf{is}  utilized to prove the  Kakutani's theorem \cite{kakutani1941generalization}\onlyConf{~which in turn is required to show the existence of NEP solution}.

\onlyJournalCand{
    We have already shown 
  that this game is similar to the routing
    game described in \cite{orda1993competitive}. With the average-cost based pricing and the zero-revenue model, the effective cost
    function for agent $i$ to minimize in the routing game is
\[\linkcost{i}{t} = \payment{i}{t} =
\dem{i}{t}\avgCostPrice(\demsum{}{t})=
\frac{\dem{i}{t}}{\demsum{}{t}}\totalFuelCost(\demsum{}{t}) .\] This
cost function satisfies the assumptions G1-G5 given earlier. In
particular, G1 holds as the total payment made by the
consumers satisfies \[\payment{i}{} = \sum_t\payment{i}{t}, \] which is the cost
to the agents in the routing formulation. In addition, G2 trivially holds by the definition
of $\payment{i}{t}$. In order to satisfy G3, i.e., to show that
$\linkcost{i}{t}$ is convex over $\dem{i}{t}$, we show that
$\frac{\partial^2 \linkcost{i}{t}}{\partial {\dem{i}{t}}^2}\geq
0$. First we evaluate
\begin{equation}
  \label{eq:mo1}
  \begin{split}
    \frac{\partial \linkcost{i}{t}}{\partial \dem{i}{t}} &=
    \frac{\partial\left( \frac{\dem{i}{t}}{\demsum{}{t}}\totalFuelCost(\demsum{}{t})\right)}{\partial \dem{i}{t}} \\
    & = \ldots\\
    &= \frac{(\demsum{}{t}- \dem{i}{t})\avgCostPrice(\demsum{}{t}) +
      \dem{i}{t}\totalFuelCost'(\demsum{}{t}) }{\demsum{}{t}}.
    \\
  \end{split}
\end{equation}
Then we evaluate 
\begin{equation}
  \label{eq:mot2}
  \begin{split}
    \frac{\partial^2 \linkcost{i}{t}}{\partial {\dem{i}{t}}^2}
    &= \frac{1}{{\demsum{}{t}}^2}\left[ 2(\demsum{}{t}-
      \dem{i}{t})\left(\totalFuelCost'(\demsum{}{t}) -
        \frac{\totalFuelCost(\demsum{}{t})}{\demsum{}{t}}\right)\right.\\&\quad ~+
      \demsum{}{t}\dem{i}{t}\totalFuelCost''(\demsum{}{t}) \Big].
  \end{split}
\end{equation}
Given $\totalFuelCost(x)$ is convex, both
$\left(\totalFuelCost'(\demsum{}{t}) -
  \frac{\totalFuelCost(\demsum{}{t})}{\demsum{}{t}}\right)\geq 0 $ and
$\totalFuelCost''(\demsum{}{t})\geq0$ \cui{hold}; and therefore $\frac{\partial^2
  \linkcost{i}{t}}{\partial {\dem{i}{t}}^2}\geq 0$. Thus
$\linkcost{i}{t}$ is convex over $\dem{i}{t}$ and G3 holds. The above
also shows that $\linkcost{i}{t}$ is continuously differentiable over
$\dem{i}{t}$ and hence G4 holds. Finally, 
\cui{we assume that} G5 holds by construction.}
\onlyConf{
By the results in \cite{orda1993competitive} we can show that 
there exists an NEP
strategy for all agents with the cost function used here and the NEP solution exists for the \cui{proposed} 
noncooperative consumer load balancing game.
}
\onlyJournal{By the results in \cite{orda1993competitive} we know that if the cost
function satisfies assumption G2 and G3, and $\Dem{}{}$, 
is compact, 
there exists an NEP
strategy for all agents. Therefore, the NEP solution exists for the \cui{proposed} 
noncooperative consumer load balancing game.}

\onlyConf{\cui{On the other hand,} the cost function \onlyJournal{$\linkcost{i}{t}$}\onlyConf{$\payment{i}{t}$} does not satisfy the conditions for being a type-A function, \cui{defined in} \cite{orda1993competitive}.}\onlyJournal{\cui{On the other hand,} the cost function $\linkcost{i}{t}$ does not satisfy the assumption
A3; so it \cui{is disqualified} 
as a type-A function \cui{defined in}  \cite{orda1993competitive}.} \cui{Therefore,} 
the corresponding uniqueness result \cui{in \cite{orda1993competitive}} cannot be extended to our
formulation. \onlyJournal{\cui{Next,} we prove the uniqueness of the NEP solution by extending
the result in \cite{bhaskar2009equilibria}. 
The number of player \emph{types} in a
particular game refers to the number of different values for
$\demConstraint{i}$'s. Thus, a single type of players implies that all
players have the same value for $\demConstraint{i}$'s.  Next, we first
introduce some definitions and show that 
\cui{our} load balancing problem satisfies the conditions for
NEP uniqueness as described in \cite{bhaskar2009equilibria}. 
We now begin with some definitions from \cite{bhaskar2009equilibria}.



\begin{defi}
  The \emph{component-join} operation for two given graphs $G_1 = (V_1
  , E_1 )$ and $G_2 = (V_2 , E_2 )$ consists of \emph{merging} any two
  vertices $v_1 \in V_1$ and $v_2 \in V_2$ into a single vertex
  $v$. 
\end{defi}

\begin{defi}
\cuit{  Consider two nodes, named \emph{hubs}, then a \emph{hub-component}
  is a connected graph formed by connecting edges and nodes to the hubs such
  that all paths between the hubs remain vertex-disjoint.}

\end{defi}

  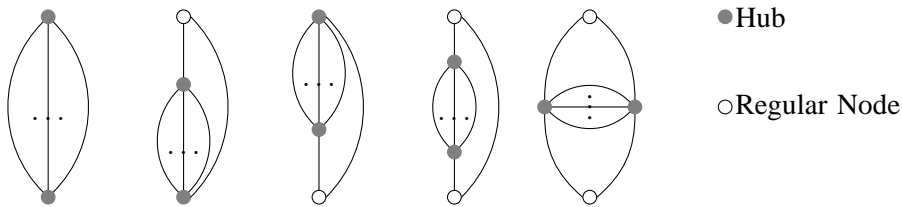
\begin{figure}[!h]
    \centering
    \tikzstyle{white} = [draw, circle, color=black, scale=0.5]
    \tikzstyle{black} = [draw, circle, color=gray, fill, scale=0.5]
\subfigure{
    \begin{tikzpicture}[>=stealth, scale=0.3]
      \node (n11) at (0,8) [black] {};
      \node (n12) at (0,0) [black] {};
      \draw (n11) to [bend left=45] (n12);
      \draw (n11) to [bend right=45] (n12);
      \draw (n11) -- (n12) node[midway, below] {$\ldots$};

      \node (n21) at (6,8) [white] {};
      \node (n22) at (6,5) [black] {};
      \node (n23) at (6,0) [black] {};
      \draw (n22) to [bend left=45] (n23);
      \draw (n22) to [bend right=45] (n23);
      \draw (n22) -- (n23) node[midway, below] {$\ldots$};
      \draw (n21) -- (n22);
      \draw (n21.east) to [bend left=45] (n23.east);

      \node (n31) at (12,8) [black] {};
      \node (n32) at (12,3) [black] {};
      \node (n33) at (12,0) [white] {};
      \draw (n31) to [bend left=45] (n32);
      \draw (n31) to [bend right=45] (n32);
      \draw (n31) -- (n32) node[midway, below] {$\ldots$};
      \draw (n32) -- (n33);
      \draw (n31.east) to [bend left=45] (n33.east);

      \node (n41) at (18,8) [white] {};
      \node (n42) at (18,6) [black] {};
      \node (n43) at (18,2) [black] {};
      \node (n44) at (18,0) [white] {};
      \draw (n42) to [bend left=45] (n43);
      \draw (n42) to [bend right=45] (n43);
      \draw (n42) -- (n43) node[midway, below] {$\ldots$};
      \draw (n41) -- (n42);
      \draw (n43) -- (n44);
      \draw (n41.east) to [bend left=45] (n44.east);

      \node (n51) at (24,8) [white] {};
      \node (n52) at (22,4) [black] {};
      \node (n53) at (26,4) [black] {};
      \node (n54) at (24,0) [white] {};
      \draw (n52) to [bend left=45] (n53);
      \draw (n52) to [bend right=45] (n53);
      \draw (n52) -- (n53) node[midway,yshift=3pt] {$\vdots$};
      \draw (n51.west) to [bend right=25] (n52.north);
      \draw (n51.east) to [bend left=25] (n53.north);
      \draw (n52.south) to [bend right=25] (n54.west);
      \draw (n53.south) to [bend left=25] (n54.east);
\onlyJournal{
\node (n61) at (30,8) [black] {};
\draw (n61) node[right] {\cuit{Hub}};
\node (n62) at (30,4) [white] {};
\draw (n62) node[right] {\cuit{Regular Node}};}
\end{tikzpicture}
}
\onlyConf{
\subfigure{
  \begin{tikzpicture}
    \node (n61) at (0,0) [black] {};
    \draw (n61) node[right] {\cuit{Hub\textcolor{white}{g}}};
    \node (n62) at (1,0) [white] {};
    \draw (n62) node[right] {\cuit{Regular Node}};
  \end{tikzpicture}
}
}
\caption{\cuit{Hub-components. These five basic units are used to construct nearly-parallel graphs.}
}
    \label{fig:npn}
  \end{figure}

\begin{defi}
  A \emph{generalized nearly-parallel graph} is any graph that can be
  constructed from \emph{hub-components} applying
  \emph{component-join} operations.
\end{defi}

\cui{The} two-node network with parallel links \cui{in}  Fig.~\ref{fig:twonode} is the
first of the five basic units (as drawn in Fig.~\ref{fig:npn}) of
nearly-parallel graphs \cite{bhaskar2009equilibria,
  richman2007topological}; and by definitions a hub-component
is also a generalized nearly-parallel graph.

\begin{defi}
  A cost function $f(x)$ is (strictly) \emph{semi-convex} if
  $xf(x)$ is (strictly) convex.
\end{defi}

For the load balancing game with average-cost based pricing and zero
revenue, the cost of consumer $i$ is given by (\ref{eq:cost1}) where
the \cui{price} 
$\avgCostPrice(\demsum{}{t})$ is given by
(\ref{eq:2}), with $\avgCostPrice(\demsum{}{t})$ a non-negative and
non-decreasing function. For the function
$\avgCostPrice(\demsum{}{t})$ to be strictly semi-convex,
$\demsum{}{t}\avgCostPrice(\demsum{}{t})$ needs to be strictly
convex. Since $\totalFuelCost(\demsum{}{t})
=\demsum{}{t}\avgCostPrice(\demsum{}{t})$ is the total cost of
electricity to the retailer, and we assume \cui{that}  the marginal cost price
$\totalFuelCost'(\demsum{}{})$ is a monotonically increasing function,
$\totalFuelCost(\demsum{}{})$ is strictly convex. }\onlyConf{In \cite{dummy2011} we show that our problem }\onlyJournal{Thus our problem }\cui{is equivalent to}  
an atomic flow game \onlyConf{\cite{bhaskar2009equilibria}} with
splittable flows and different \emph{player types} (i.e., each player
controls a different amount of total flow) over a generalized
nearly-parallel graph, \cui{which has} strictly semi-convex, non-negative, and
non-decreasing functions for cost per unit flow. By
the results of \cite{bhaskar2009equilibria}, \onlyConf{we can prove that} the NEP solution for the
load balancing game is unique\onlyConf{~\cite{dummy2011}}.

In the following, \cui{we discuss the properties for} 
the unique NEP \cui{solution for the proposed load balancing game.} 

\begin{lem}
\label{lem1}
With the average-cost based pricing and zero revenue, at the Nash
equilibrium the 
price of electricity faced by all consumers 
\cui{is the} same over all time slots.
\end{lem}
\onlyConf{The proof is provided in \cite{dummy2011}.}
\onlyJournal{\begin{proof}
  Consider two arbitrary time slots $t_1$ and $t_2$. At the Nash
  equilibrium the sum demands in the system \cui{over the two time slots} are either the same 
  or different. If the sum demands are equal over the
  two time slots, by (\ref{eq:2}), we know that the 
  price of
  electricity will be same \cui{over} 
  the two slots. If the sum demands are
  not equal, without losing generality, let us assume $\demsum{}{t_1} <
  \demsum{}{t_2}$ such that
  \begin{equation}
    \avgCostPrice(\demsum{}{t_1}) < \avgCostPrice(\demsum{}{t_2})
    \label{costcompar}
  \end{equation}
  holds. Then any consumer $j$ with $\dem{j}{t_2}> 0$ can reduce cost by
  reducing $\dem{j}{t_2}$ and increasing $\dem{j}{t_1}$ by the same
  small quantity.  This contradicts our assumption that the system
  is in equilibrium. Hence $\avgCostPrice(\demsum{}{t_1}) =
  \avgCostPrice(\demsum{}{t_2})$.
\end{proof}
}
\begin{lem}
  If $\totalFuelCost(\cdot)$ is strictly convex, at the Nash equilibrium,
  the sum of demands on the system, $\demsum{}{t}$, keeps the same
  across \cui{different time slots}.
\end{lem}
\onlyConf{The proof is provided in \cite{dummy2011}.}
\onlyJournal{\begin{proof} In order to prove this we will show that $\demsum{}{t_1}
  \neq \demsum{}{t_2}$ leads to contradiction. 
  At Nash equilibrium we have
  $\avgCostPrice(\demsum{}{t_1}) = \avgCostPrice(\demsum{}{t_2})$ from
  Lemma~\ref{lem1} for all possible $t_1$ and $t_2$. \cuit{ Thus, we obtain
  \begin{equation*}
    \begin{split}
      \avgCostPrice(\demsum{}{t_1}) &=
      \avgCostPrice(\demsum{}{t_2})\\
      \Leftrightarrow \totalFuelCost(\demsum{}{t_1})/\demsum{}{t_1} & = \totalFuelCost(\demsum{}{t_2})/\demsum{}{t_2}\\
      \Leftrightarrow \totalFuelCost(\demsum{}{t_1}) & = \totalFuelCost(\demsum{}{t_2})\frac{\demsum{}{t_1}}{\demsum{}{t_2}}.
    \end{split}
  \end{equation*}
  If
  $\totalFuelCost(\cdot)$ is strictly convex and
  $\demsum{}{t_1} \neq \demsum{}{t_2}$ (without loss of generality say $\demsum{}{t_1} < \demsum{}{t_2}$), by definition of strict convexity, we have
\begin{equation*}
  \begin{split}
     \totalFuelCost(\demsum{}{t_1}) &<\frac{\demsum{}{t_2} - \demsum{}{t_1}}{\demsum{}{t_2}}  \totalFuelCost(0) + \frac{\demsum{}{t_1}}{\demsum{}{t_2}}  \totalFuelCost(\demsum{}{t_2})\\
\Rightarrow \totalFuelCost(\demsum{}{t_2})\frac{\demsum{}{t_1}}{\demsum{}{t_2}} &<\frac{\demsum{}{t_2} - \demsum{}{t_1}}{\demsum{}{t_2}}  \times 0 + \frac{\demsum{}{t_1}}{\demsum{}{t_2}}  \totalFuelCost(\demsum{}{t_2})\\
\Rightarrow \totalFuelCost(\demsum{}{t_2})\frac{\demsum{}{t_1}}{\demsum{}{t_2}} &< \frac{\demsum{}{t_1}}{\demsum{}{t_2}}  \totalFuelCost(\demsum{}{t_2})\\
  \end{split}
\end{equation*}
which is a contradiction and hence $\demsum{}{t_1} = \demsum{}{t_2}$.
}
\end{proof}
}


\begin{lem}
  If $\totalFuelCost(\cdot)$ is strictly convex, at Nash equilibrium, each
  consumer will distribute its demands equally over the
  $\numTimeSlots$ time slots.
\end{lem}
\onlyConf{The proof is provided in \cite{dummy2011}.\\\\}
\onlyJournal{\begin{proof}
  As the Nash equilibrium is unique, by symmetry over all the time
  slots, for consumer $i$ we shall have \[\dem{i}{t_1} = \dem{i}{t_2},
  \quad \forall t_1, t_2,\] as otherwise we could swap the demand
  vectors $\dem{}{t_1}$ and $\dem{}{t_2}$ in time slots $t_1$ and $t_2$
  without altering the Nash equilibrium conditions and get another
  distinct NEP, thus contradicting \cui{the} uniqueness.  
Thus, with
  \cuit{$\sum_t\dem{i}{t}\geq\demConstraint{i}$ and the fact that the consumer is trying to minimize their cost 
such that
    $\sum_t\dem{i}{t}=\demConstraint{i}$ holds at equilibrium,} we have the solution
  $\dem{i}{t}=\demConstraint{i}/\numTimeSlots$ for all consumers $i$ and
  time slots $t$.
\end{proof}
}
{\bf \noindent Remark:} Under the average-cost based pricing scheme with zero revenue, if one
particular consumer increases its \cui{total} demand of electricity, the 
price $\avgCostPrice(\cdot)$ 
increases, which in turn increases the payments
for all other consumers as well. Theoretically one consumer may cause
indefinite increases in the payments of all others; and in this sense
this scheme does not protect the group from reckless action of some
consumer(s). This issue will be addressed by our second pricing scheme
as we will show in Section \ref{sec:new-pricing-scheme}.


\subsection{Constant-Rate Revenue Model}
\label{sec:const-rate-reve-mod1}

In this case, the rate of revenue generation for each consumer at each time slot is
taken as a non-negative constant $\rateofrevenue{i}{t}$.
Thus,
\onlyJournal{\[\revenue{i}{t} = \rateofrevenue{i}{t}\times \dem{i}{t}.\]}\onlyConf{$\revenue{i}{t} = \rateofrevenue{i}{t}\times \dem{i}{t}$.} 
The consumer 
\cui{load balancing}
problem for each consumer $i$ is given by
the following optimization problem:
      \begin{equation*}
      \begin{split}
        \text{maximize}\quad & \payoff{i}{}(\dem{i}{})=
        \sum_t\left(\revenue{i}{t}-\payment{i}{t}\right)\\
        \text{subject to}\quad 
        &\revenue{i}{t} = \rateofrevenue{i}{t}\dem{i}{t},\quad\forall t,\\
        &\payment{i}{t} = \dem{i}{t}\avgCostPrice(\demsum{}{t}),\quad\forall t,\\
        & \sum_t\dem{i}{t}
        \geq\demConstraint{i},\\
        & \demsum{}{t}=\sum_j\dem{j}{t},\quad\forall t,
\\
        &0\leq \dem{i}{t},\quad\forall  t.\\
      \end{split}
    \end{equation*}


\cui{We assume that} $\demConstraint{i}=0,~\forall i,$ and the 
    rate of
    revenue \cui{is larger} 
    than the 
    price of electricity 
\cui{such that we do not end up with} any negative payoff or 
    the trivial solution
    $\dem{i}{t}=0,~\forall i,t$.  

    Here again, if the sum demand in a given time slot $t$ exceeds the
    \cui{retailer's} capacity threshold $\maxWholeSaleCapacity$,
     the consumers will face an infinite price for their
    consumption. 
    This implies that, at
    Nash equilibrium \cui{the} sum demand $\demsum{}{t}$ will never exceed the
    capacity threshold $\maxWholeSaleCapacity$, as \onlyJournal{G5 holds}\onlyConf{we assume that sum load demand over all time slots is greater that sum load available}. 
    This
    again allows for the redundant constraint
    $\dem{i}{t}\leq\maxWholeSaleCapacity, ~\forall i, t$, as
    $\dem{i}{t}\leq\sum_i\dem{i}{t}=\demsum{}{t}\leq\maxWholeSaleCapacity$,
    which in turn makes the feasible region for $\dem{}{}$,
    $\Dem{}{}$, finite and hence compact. 

\onlyConf{The proof for the existence of NEP for this game under the
given assumptions is provided in \cite{dummy2011}.}
\onlyJournalCand{
    We briefly show that under these assumptions there exists an NEP
    for this game. \cui{In particular}, the effective cost function for the corresponding
    routing game is given as
\begin{equation}
\linkcost{i}{t}=\payment{i}{t} - \rateofrevenue{i}{t}\times
\dem{i}{t}\label{eq:ng13}.
\end{equation}
As $\payment{i}{t}$ is continuous in $\dem{}{t}$, $\linkcost{i}{t}$ is
continuous in $\dem{}{t}$ as well and satisfies assumption G2.  We
have already shown that $\payment{i}{t}$ under the average-cost based
pricing scheme is convex in $\dem{i}{t}$ through (\ref{eq:mot2}). The
function $- \rateofrevenue{i}{t} \dem{i}{t}$ is linear and hence
convex in $\dem{i}{t}$. Thus, by the property that the summation of
two convex functions is convex, $\linkcost{i}{t}$ from (\ref{eq:ng13})
is convex in $\dem{i}{t}$ and hence satisfies assumption
G3. Following the proof in \cite{orda1993competitive}, we consider the
point-to-set mapping $\dem{}{}\in \Dem{}{}\rightarrow
\Gamma(\dem{}{})\subset \Dem{}{}$ defined as
\begin{equation}
\Gamma(\dem{}{})
=\{\demt{}{}\in\Dem{}{}:\demt{i}{}\in\arg\min_{\demb{i}{}\in\Dem{i}{}}\linkcost{i}{}(\dem{1}{},\ldots,
\demb{i}{}, \ldots, \dem{\numPlayer}{})\} ,\label{eq:neweq1}
\end{equation}
where $\Gamma$ is an
upper semicontinuous mapping (by the continuity assumption G2) that
maps each point of the convex compact set $\Dem{}{}$ into a closed (by
G2) convex (by G3) subset of $\Dem{}{}$. By the Kakutani Fixed Point
Theorem \cite{kakutani1941generalization}, there exists a fixed point
$\dem{}{}\in\Gamma(\dem{}{})$ and such a point is an NEP \cite{nash1951non}.



}
\begin{lem}\label{p12}
  At the Nash equilibrium, the consumer(s) with the highest revenue
  rate ($\rateofrevenue{i}{t}$) within the time slot, may be the only
  one(s) buying the power in that time slot.
\end{lem}
\onlyConf{The proof is provided in \cite{dummy2011}. }\onlyJournal{\begin{proof}
  For a given time slot $t$ consumer $i$ has an incentive to
  increase its demand $\dem{i}{t}$ as long as the payoff can increase,
  i.e., as long as \onlyJournal{\[\frac{\partial \pi_i}{\partial \dem{i}{t}}> 0.\]}\onlyConf{$\partial \pi_i/\partial \dem{i}{t}> 0$.} 
  Therefore at the equilibrium the following holds for all consumers.
\begin{equation*}
  \begin{split}
    \frac{\partial \payoff{i}{}}{\partial \dem{i}{t}} &\leq 0 \\
    \Rightarrow \frac{\partial \left(\revenue{i}{t} - \payment{i}{t} \right) }{\partial \dem{i}{t}} &\leq 0\\
    \Rightarrow\frac{\partial \revenue{i}{t} }{\partial
      \dem{i}{t}} - \frac{\partial \payment{i}{t}  }{\partial \dem{i}{t}} &\leq 0\\
    \Rightarrow \rateofrevenue{i}{t} &\leq \frac{\partial \payment{i}{t}  }{\partial \dem{i}{t}} = \frac{\totalFuelCost(\demsum{}{t})}{\demsum{}{t}}=A(\demsum{}{t})\\
  \end{split}
\end{equation*}
For the consumers with a strict inequality $\rateofrevenue{i}{t} <
A(\demsum{}{t})$, the rate of revenue is less than the price per unit
of electricity at time $t$; hence the revenue is less than the cost,
$\revenue{i}{t} < \payment{i}{t}$, such that buying electricity will
incur them a negative payoff and hence all such consumers, with
$\rateofrevenue{i}{t} < \avgCostPrice(\demsum{}{t})$, will not buy any
power in that time slot, i.e., $\dem{i}{t}=0$.  Therefore only the set
of consumers $\{\arg\max_k \rateofrevenue{k}{t}\}$, i.e., the
consumers who enjoy the maximum rate of revenue may be able to
purchase electricity.
\end{proof}}Thus if consumer $i$ has the maximum rate of revenue, either it is the
only consumer buying non-zero power $\dem{i}{t}$ such that
$\rateofrevenue{i}{t}=\avgCostPrice(\demsum{i}{t})$ or
$\rateofrevenue{i}{t}<\totalFuelCost'(0)$ and hence $\dem{i}{t}=0$
in that time slot, which leads to a unique Nash equilibrium for the
sub-game. If in a given time slot multiple consumers experience the
same maximum rate of revenue, the sub-game will turn into a Nash
Demand Game \cite{nash1953two} between the set of consumers given by
$\{\arg\max_k \rateofrevenue{k}{t}\}$, which is well known to admit
multiple Nash equilibriums. Thus the overall noncooperative game
has a unique Nash equilibrium if and only if, in each time slot, at
most one consumer experiences the maximum rate of revenue.

\section{Nash Equilibrium with Increasing-Block Pricing}
\label{sec:new-pricing-scheme}

In this section we study the load balancing game with the
time-variant increasing-block pricing scheme.
Under this scheme consumer $i$ pays $\payment{i}{t}$ for $\dem{i}{t}$
units of electricity, which is given by
(\ref{cost2}) \cui{with} $\IBP(\demandSymbol, \dem{}{t})$  the marginal cost
function posed to the consumer
. Thus, \cui{as defined before, we have}

\[\IBP(\demandSymbol, \dem{}{t}) =  \totalFuelCostMCP\left(\sum_j \min{(\demandSymbol,\dem{j}{t})}\right)
 .\]
 As an example, if the \cui{demands from different consumers} at time slot $t$ are
 identical, i.e., if $\dem{i}{t}=\dem{j}{t},~\forall i,j$, we
 have, \[\IBP(\demandSymbol,\dem{}{t})=\totalFuelCostMCP(\numPlayer
 \demandSymbol). \]




 \subsection{Zero-Revenue Model}
 \label{sec:zero-reve-model}

In this case the payment by consumer $i$ is given by (\ref{cost2})
\[\payment{i}{t}=\int_0^{\dem{i}{t}}\IBP(\demandSymbol, \dem{}{t})d\demandSymbol.\]
The consumer \cui{load balancing} 
problem for each consumer $i$ is given by
the following optimization problem:
      \begin{equation*}
      \begin{split}
        \text{maximize}\quad & \payoff{i}{}(\dem{i}{})=
        -\sum_t\payment{i}{t}\\
        \text{subject to}\quad 
        &\payment{i}{t} = \int_0^{\dem{i}{t}}\IBP(\demandSymbol, \dem{}{t})d\demandSymbol,\quad\forall t,\\
        & \sum_t\dem{i}{t}
        \geq\demConstraint{i},
\\
        &0\leq \dem{i}{t},\quad\forall  t.\\
      \end{split}
    \end{equation*}
    If the sum demand $\demsum{}{t}$ in a time slot $t$ exceeds
    $\maxWholeSaleCapacity$, the price of electricity for the consumer
    with the highest demand (\cui{indexed by} $\hat{j}$) becomes infinite. As we
    retain the assumption \onlyJournal{G5}\onlyConf{that sum load demand over all time slots is greater that sum load available}, 
    consumer $\hat{j}$ can
    rearrange its demand \cui{vector such that} 
    either \cui{the} sum demand
    becomes within \cui{the} capacity threshold or 
    consumer $\hat{j}$ is no
    longer the highest demand consumer \cui{(then the new customer with the highest demand performs the same routine until the sum demand is under the threshold)}. This implies that, at the Nash
    equilibrium point 
\cui{we have} 
    $\demsum{}{t}\leq\maxWholeSaleCapacity$. \cui{Similarly, we now have the} 
    redundant constraint
    $\dem{i}{t}\leq\maxWholeSaleCapacity,~\forall~ i,~ t$
    ,
    which in turn makes the feasible region 
    $\Dem{}{}$ finite and hence compact. 

\onlyConf{The proof for the existence of NEP for this game under the
given assumptions is provided in \cite{dummy2011}.}
\onlyJournal{
    As $\payment{i}{t}$ is continuous in $\dem{}{t}$, 
    in the
    corresponding routing game \cui{we have that}
\begin{equation}
\linkcost{i}{t}=\payment{i}{t}\label{eq:ng14}
\end{equation}
is continuous in $\dem{}{t}$ and satisfies assumption G2.
In addition, $\payment{i}{t}$ is convex in $\dem{i}{t}$ as its derivative, the
marginal cost function $\IBP(\demandSymbol, \dem{}{t})$, is
non-decreasing. 
Thus, 
$\linkcost{i}{t}$ 
is convex in $\dem{i}{t}$ and hence satisfies  assumption
G3. Following the proof in \cite{orda1993competitive}, we consider
the point-to-set mapping $\dem{}{}\in
\Dem{}{}\rightarrow \Gamma(\dem{}{})\subset
\Dem{}{}$ defined \cui{the same as in (\ref{eq:neweq1})}. 
By the Kakutani Fixed Point
Theorem \cite{kakutani1941generalization}, there exists a fixed point
$\dem{}{}\in\Gamma(\dem{}{})$ and such a point is
 an NEP.
}
\onlyJournal{

 }When each consumer tries to minimize its total cost while satisfying
 its minimum daily energy requirement $\demConstraint{i}$, we have the
 following result.

\begin{lem}\label{p21}
  If $\totalFuelCost(\cdot)$ is strictly convex, the Nash equilibrium is
  unique and each consumer distributes its demand uniformly over all
  time slots. 
\end{lem}
\onlyConf{The proof is provided in \cite{dummy2011}.\\\\}
\onlyJournal{
\begin{proof}
  For the equilibrium conditions to be satisfied, \[
  \IBP(\dem{i}{t_1}, \dem{}{t_1}) = \IBP(\dem{i}{t_2},
  \dem{}{t_2}),\quad \forall i,t_1, t_2,\] should hold; 
  otherwise consumer $i$ can increase payoff by varying $\dem{i}{t_1}$
  and $\dem{i}{t_2}$, \cui{in a similar argument to that for Lemma~\ref{lem1}}. This condition can be rewritten after expanding
  $\IBP(\cdot)$ as
  \begin{equation}
    \begin{split}
      \totalFuelCostMCP\left(\sum_j
        \min{(\dem{i}{t_1},\dem{j}{t_1})}\right) &=
      \totalFuelCostMCP\left(\sum_j
        \min{(\dem{i}{t_2},\dem{j}{t_2})}\right) ,\quad \forall
      i,t_1, t_2.
    \end{split}
    \label{eq:prf42}
  \end{equation}
  Given \cui{that} $\totalFuelCost(\cdot)$ is strictly convex, we have
  $\totalFuelCostMCP(\cdot) = \totalFuelCost'(\cdot)$ monotonically increasing, 
  which gives
  \begin{equation}
\totalFuelCostMCP(\demsum[z]{}{1}) =\totalFuelCostMCP(\demsum[z]{}{2}) \Leftrightarrow \demsum[z]{}{1}=\demsum[z]{}{2}\label{eq:3}.
\end{equation}
Therefore, (\ref{eq:prf42}) implies
  \begin{equation}
    \begin{split}
      \sum_j\min{(\dem{i}{t_1},\dem{j}{t_1})} &=
      \sum_j\min{(\dem{i}{t_2},\dem{j}{t_2})},\quad\forall i, t_1, t_2.\\
    \end{split}
    \label{eq:lem4}
  \end{equation}
  Now 
  assume that there exists an NEP $\dem{}{}$ with demand
  vectors $\dem{}{t_1}\neq \dem{}{t_2}$. Let $\mathcal{P}$ represent the subset of consumers with
  unequal demands in time slots $t_1$ and $t_2$, \cui{i.e.,}
  \[\mathcal{P}=\{k| \dem{k}{t_1}\neq
  \dem{k}{t_2}, k\in \{1,2,\ldots,\numPlayer\}\} .\] Then let $a$
  represent the consumer from subset $\mathcal{P}$ with the highest value of
  demand in time slot $t_1$, \cui{i.e.,}
  \begin{equation}
    \label{eq:4}
    a = \arg\max_{k\in\mathcal{P}} \dem{k}{t_1},
  \end{equation}
  and let $b$ represent the consumer from subset $\mathcal{P}$ with
  the highest value of demand in time slots $t_2$, \cui{i.e.,}
  \begin{equation}
    \label{eq:5}
    b =    \arg\max_{k\in\mathcal{P}} \dem{k}{t_2} .
  \end{equation}
  From (\ref{eq:lem4}) we have
\begin{equation}
  \label{eq:6}
   \sum_j\min{(\dem{a}{t_1},\dem{j}{t_1})} =
      \sum_j\min{(\dem{a}{t_2},\dem{j}{t_2})},\quad\forall t_1, t_2,
\end{equation}
and
\begin{equation}
  \label{eq:7}
   \sum_j\min{(\dem{b}{t_1},\dem{j}{t_1})} =
      \sum_j\min{(\dem{b}{t_2},\dem{j}{t_2})},\quad\forall t_1, t_2.
\end{equation}
Combining (\ref{eq:lem4}) and (\ref{eq:4}) leads to
\begin{equation}
  \label{eq:8}
  \sum_j\min{(\dem{a}{t_1},\dem{j}{t_1})} \geq \sum_j\min{(\dem{b}{t_1},\dem{j}{t_1})};
\end{equation}
combining (\ref{eq:lem4}) and (\ref{eq:5}) leads to
\begin{equation}
  \label{eq:9}
  \sum_j\min{(\dem{a}{t_2},\dem{j}{t_2})} \leq
  \sum_j\min{(\dem{b}{t_2},\dem{j}{t_2})} .
\end{equation}
If $\dem{a}{t_1}\neq \dem{b}{t_1}$ or $\dem{a}{t_2}\neq \dem{b}{t_2}$,
(\ref{eq:8}) holds with strict inequality. With 
(\ref{eq:6}), (\ref{eq:7}), and (\ref{eq:8}), we have
\[\sum_j\min{(\dem{a}{t_2},\dem{j}{t_2})}>
\sum_j\min{(\dem{b}{t_2},\dem{j}{t_2})},\]
which contradicts (\ref{eq:9})
. If $\dem{a}{t_1} = \dem{b}{t_1}$ and $\dem{a}{t_2}=
\dem{b}{t_2}$, 
(\ref{eq:6}) and (\ref{eq:7}) imply $\dem{a}{t_1} =
\dem{a}{t_2}$ and $\dem{b}{t_1}= \dem{b}{t_2}$, respectively, which contradicts
that $a,b\in \mathcal{P}$. This implies that the set $\mathcal{P}$ is
empty, which contradicts that $\dem{}{t_1}\neq
 \dem{}{t_2}$.

 Hence we have \[\dem{i}{t_1} = \dem{i}{t_2} \qquad \forall i,t_1 ,
 t_2,\] and the solution is given by $\dem{i}{t} =
 \demConstraint{i}/\numTimeSlots,~ \forall i,t$. Under the necessary
 conditions for NEP (\ref{eq:prf42}), this is the only solution for
 the set $\dem{}{}$, hence NEP is unique.
\end{proof}
}
{\bf \noindent Remark:} Notice that \cui{under the} zero-revenue model, the 
NEP \cui{point} is the
same with both increasing-block pricing and average-cost based
pricing. For both the cases, at NEP, we have $\dem{i}{t} =
\demConstraint{i}/\numTimeSlots,~ \forall i,t$. However, even though
the loading pattern is similar, the payments $\payment{i}{t}$ made by
the consumers will differ and, with increasing-block pricing, \cui{will} 
likely 
be lesser for consumers with relatively \cui{lower} 
consumption. 
\cui{In addition, with increasing-block pricing},
the maximum payment $\payment{i}{t}$ made by
the $i$th consumer given $\dem{i}{t}$ demand will be
$\totalFuelCost(\numPlayer \dem{i}{t})/\numPlayer$, irrespective of
what other consumers demand and consume. Thus this addresses the issue
faced \cui{under the} 
average-cost based pricing and zero-revenue model, in which one
consumer can increase their demand indefinitely and cause indefinite
increase in the payments of all other consumers.


\subsection{Constant-Rate Revenue Model}
\label{sec:const-rate-reve-mod2}

The consumer \cui{load balancing} 
problem for 
consumer $i$ is given by
the following optimization problem:
      \begin{equation*}
      \begin{split}
        \text{maximize}\quad & \payoff{i}{}(\dem{i}{})=
        \sum_t\left(\revenue{i}{t}-\payment{i}{t}\right)\\
        \text{subject to}\quad &\revenue{i}{t} = \rateofrevenue{i}{t}
        \dem{i}{t},\quad\forall t,\\
        &\payment{i}{t} = \int_0^{\dem{i}{t}}\IBP(\demandSymbol, \dem{}{t})d\demandSymbol,\quad\forall t,\\
        & \sum_t\dem{i}{t}
        \geq\demConstraint{i},
\\
        &0\leq \dem{i}{t},\quad\forall  t.\\
      \end{split}
    \end{equation*}

    Here again, we assume  $\demConstraint{i}=0,~\forall i$, to avoid any negative payoffs and \cui{we could agree 
    for the redundant constraint}
    $\dem{i}{t}\leq\maxWholeSaleCapacity,~\forall~ i,~ t$, 
    which in turn makes the feasible region for 
    $\Dem{}{}$ finite and hence compact. 

\onlyConf{The proof for the existence of NEP for this game under the given assumptions is provided in \cite{dummy2011}.}
\onlyJournal{
    In this case, we briefly show that with increasing-block pricing
    and a constant-rate for revenue, there exists an NEP solution for this
    game. First, the cost function in the corresponding routing game
    is given by
\begin{equation}
\linkcost{i}{t}=\payment{i}{t} - \rateofrevenue{i}{t}\times
  \dem{i}{t}\label{eq:ng15},
\end{equation}
where $\payment{i}{t}$ is continuous in $\dem{}{t}$, and therefore
$\linkcost{i}{t}$ is continuous in $\dem{}{t}$ and satisfies
assumption G2.  We have already shown that $\payment{i}{t}$ under
the increasing-block pricing scheme is convex in $\dem{i}{t}$ in \cui{the} previous
subsection. The function $- \rateofrevenue{i}{t} \dem{i}{t}$ is linear
and hence convex in $\dem{i}{t}$ as well. Thus, 
$\linkcost{i}{t}$ from
(\ref{eq:ng15}) is convex in $\dem{i}{t}$ and hence satisfies
assumption G3. 
\cui{By the same point-to-set mapping argument as that for Lemma~\ref{lem1}, we can have that there exists
 a fixed point $\dem{}{}\in\Gamma(\dem{}{})$ and such
 a point is NEP.}

}\cui{With} the average-cost based pricing scheme \cui{under the constant-rate revenue} 
model, we see that in a given time slot, if a single consumer enjoys
the maximum rate of revenue, it will be the only consumer who \cui{is} able to purchase power.  We show here that with the increasing-block
pricing scheme under constant-rate revenue model, the 
\cui{result is different}.

For a given time slot $t$, consumer $i$ has an incentive to
increase their demand $\dem{i}{t}$ as long as the payoff increases,
i.e., \onlyJournal{\[\frac{\partial \payoff{i}{}}{\partial \dem{i}{t}}> 0.\]}\onlyConf{$\partial \payoff{i}{}/\partial \dem{i}{t}> 0$.} 
Therefore at the equilibrium the following holds for all consumers:
\begin{equation}
  \label{eq:10s22a}
  \begin{split}
    \frac{\partial \payoff{i}{}}{\partial \dem{i}{t}} &\leq 0 \\
    \Rightarrow \rateofrevenue{i}{t} &\leq \frac{\partial \payment{i}{t}  }{\partial \dem{i}{t}} = \IBP(\dem{i}{t}, \dem{}{t}).
  \end{split}
\end{equation}
Additionally, if $\rateofrevenue{i}{t} < \IBP(\dem{i}{t}, \dem{}{t})$,
 $\linkcost{i}{t}$ can be reduced by reducing $\dem{i}{t}$. This
implies that if $\dem{i}{t}>0$, at the equilibrium we have
\begin{equation}
  \label{eq:10s22b}
  \rateofrevenue{i}{t} \geq \IBP(\dem{i}{t}, \dem{}{t}).
\end{equation}
Thus (\ref{eq:10s22a}) and (\ref{eq:10s22b}) together imply \cui{that}, if
$\dem{i}{t}>0$,  we have \onlyJournal{\[\rateofrevenue{i}{t} = \IBP(\dem{i}{t},
\dem{}{t}).\]}\onlyConf{$\rateofrevenue{i}{t} = \IBP(\dem{i}{t},
\dem{}{t})$.} Together we can write the following set of necessary
conditions for equilibrium,
\begin{equation}
\begin{array}{rlr}
    \rateofrevenue{i}{t} &= \IBP(\dem{i}{t}, \dem{}{t})& \text{if}~\rateofrevenue{i}{t} \geq \IBP(0, \dem{}{t}),\\
    \dem{i}{t} &=0 & \text{if}~\rateofrevenue{i}{t} < \IBP(0, \dem{}{t}).
\end{array}
\label{eq:neweq2}
\end{equation}

For illustration, we simulate a scenario consisting of 100 consumers,
who have their rate of revenue $\rateofrevenue{i}{t}$ generated from a
uniform distribution ranging \cui{over} $\$0-\$100/\text{MWh}$, where the
marginal cost to the retailer $\totalFuelCostMCP(\cdot)$ is given by
Fig.~\ref{fig:mcgs1}. In Fig.~\ref{fig:simResults} we plot the demand
$\dem{i}{t}$ versus the rate of revenue ($\rateofrevenue{i}{t}$) \cui{at a given time slot $t$, where $\dem{i}{t}$ is evaluated over} $i=\{1,\ldots,100\}$.
The 
equilibrium 
is obtained
by iterative updates of
$\IBP(\cdot)$ and $\dem{}{t}$ until convergence within an error
tolerance \cui{as in (\ref{eq:neweq2})}.

\begin{figure}[htp]
  \centering
  \begin{tikzpicture}[only marks, scale = 0.5, domain=0:12, samples=100]
    \draw[thick,color=gray,step=4cm, dashed] (0,0) grid (12,12);
        \draw[->] (-1,0) -- (12.5,0) node[midway, sloped, below]
{Rate of Revenue $\rateofrevenue{i}{t}$ (\$/MWh)};
\draw[->] (0,-1) -- (0,12.5)
node[midway, sloped, above]         {Quantity Demanded $\dem{i}{t}$ (MWh)};
    \draw plot[mark=*, mark size=2] file {simulation.results.table};
      \draw (0, 12) node[left] {1500};
      \draw (12, 0) node[below] {\$100};

  \end{tikzpicture}
  \caption{Demand $\dem{i}{t}$ versus the rate of revenue
    ($\rateofrevenue{i}{t}$) at equilibrium. Each dot represents a
    particular consumer $i=\{1,\ldots,100\}$.}
  \label{fig:simResults}
\end{figure}
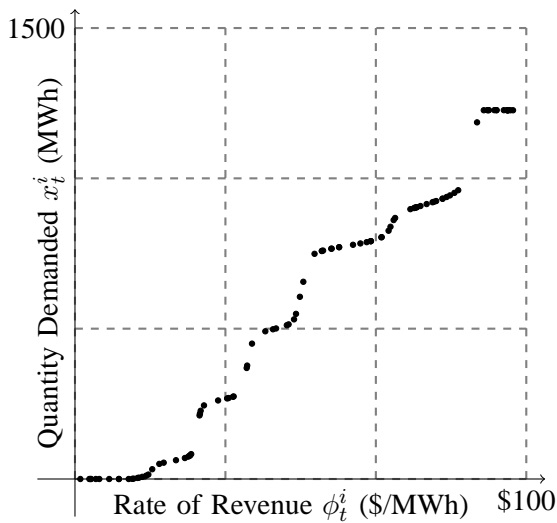

Thus, unlike with the average-cost pricing, where only the consumer with the
maximum 
rate of revenue could purchase electricity at
the equilibrium, any consumer 
\cui{may}
procure a non-zero amount of energy as long
as its own rate of revenue is larger than $\IBP(0, \dem{}{t})$.


\section{Conclusion}
\label{sec:conclusion}

In this paper we formulated noncooperative games among the consumers
of Smart Grid with two real-time pricing schemes to derive autonomous
load balancing \cui{solutions}
. The first pricing scheme charges consumers a
price that is equal to the average cost of electricity borne by the
retailer and the second scheme charges consumers \cui{an amount that is dependent on the incremental marginal cost which is shown to protect consumers from irrational behaviors.} 
Two revenue models were considered for each
of the pricing schemes, \cui{for which} 
we investigated 
the
Nash equilibrium operation points for their uniqueness and load
balancing properties. For the zero-revenue model, we showed that when
consumers are interested only in the minimization of electricity costs,
the Nash equilibrium point is unique \cui{with} both the pricing schemes and
leads to similar electricity loading patterns in \cui{both} 
cases. For the
constant-rate revenue 
model, we showed \cui{the} existence of Nash equilibrium \cui{with} 
both the pricing schemes and 
the uniqueness results \cui{with} 
the
average-cost based pricing scheme. 
\onlyJournal{~\cui{Throughout the paper, we utilized} 
the relationship between the load balancing games and
the atomic splittable flow games from the computer networking
community \cui{to prove the properties at the Nash Equilibrium solutions}.} 


\bibliographystyle{IEEEtran}
\bibliography{bibdata}

\end{document}